%% file: main.tex
\RequirePackage{amsmath} 
\documentclass{llncs}
\usepackage[ruled]{algorithm} \usepackage{algorithmicx} 
\usepackage{color} \usepackage{amsfonts} 
\usepackage[numbers,sort&compress,square]{natbib}
\usepackage{amssymb}
\usepackage{graphicx}
\usepackage{enumitem}

\newcommand{\pparagraph}[1]{{\bf{#1}.}} \renewcommand{\epsilon}{\varepsilon}
\newcommand{\tepsilon}{\tilde{\varepsilon}}
 \newcommand{\HJ}{\texttt{HJ}}

\newcommand{\lemref}[1]{Lemma~\ref{lem:#1}}

\newcommand{\thmref}[1]{Theorem~\ref{thm:#1}}

\newcommand{\secref}[1]{Section~\ref{sec:#1}}

\newcommand{\footnotenonumber}[1]{{\def\thempfn{}\footnotetext{\footnotesize #1}}}

\makeatletter
\renewcommand\bibsection%
{
  \section*{\refname
    \@mkboth{\MakeUppercase{\refname}}{\MakeUppercase{\refname}}}
}
\makeatother

\newif\ifnotes 
\notesfalse 

\ifnotes
\newcommand{\thang}[1]{\textcolor{blue}{{\footnotesize
#1}}\marginpar{\raggedright\tiny \textcolor{blue}{THANG}}}
\newcommand{\shikha}[1]{\textcolor{red}{{\footnotesize
#1}}\marginpar{\raggedright\tiny \textcolor{red}{SHIKHA}}}

\else \newcommand{\shikha}[1]{} \newcommand{\thang}[1]{} \fi

\definecolor{blueLink}{rgb}{0,0.2,0.8}

\newenvironment{proofof}[1]
{\par \vspace{1.5ex} \noindent \textit{Proof of #1.}}
{\par}

\input{macros.tex}

\title{Approximating $k$-Forest with Resource Augmentation: A Primal-Dual Approach
}

\author{
Eric Angel\inst{1} \and Nguyen Kim Thang\inst{1} \and Shikha Singh\inst{2}
}
\institute{
IBISC, University d'Evry Val d'Essonne, France. \\
\email{\{angel, thang\}@ibisc.univ-evry.fr}.
\and
Stony Brook University, Stony Brook, NY, USA. \\
\email{shiksingh@cs.stonybrook.edu}
}

\begin{document}

\maketitle
\footnotenonumber{
This research was supported by the ANR project  \small{OATA n\textsuperscript{o}ANR-15-CE40-0015-01} 
and the Chateaubriand Fellowship of the Office for Science \& 
Technology of the Embassy of France in the United States.
}

\pagestyle{plain}

\input{abstract}
\input{intro.tex}

\input{k-forest}
\input{conclusion.tex} %removed for the time-being
\input{ack.tex}

\bibliographystyle{splncsnat}
\bibliography{ckm}
\end{document}

%% file: macros.tex
\newcommand{\lastcorrections}%
{{
 \begin{sloppypar}
    \baselineskip -0.2in
    \tiny\bf\noindent
last corrections:\\
\end{sloppypar}
}}

\newcommand{\margincomment}[1]%
    {{%
      \marginpar{{\tiny\begin{minipage}{0.5in}
                       \begin{flushleft}
                          {#1}
                       \end{flushleft}
                       \end{minipage}
                }}
    }}

\newcommand{\ignore}[1]{}

\newenvironment{bigeqn*}{\large\begin{eqnarray*}}{\end{eqnarray*}}

\newcommand{\OPT}{\textsc{Opt}}

\newcommand{\lref}[2][]{\hyperref[#2]{#1~\ref*{#2}}}

%% file: abstract.tex
\begin{abstract}
In this paper, we study the $k$-forest problem in the model of resource
augmentation.  In the $k$-forest problem, given an edge-weighted graph
$G(V,E)$, a parameter $k$, and a set of $m$ demand pairs $\subseteq V \times V$,
the objective is to construct a minimum-cost subgraph that connects at least
$k$ demands.  The problem is hard to approximate---the best-known approximation
ratio is $O(\min\{\sqrt{n}, \sqrt{k}\})$.  Furthermore, $k$-forest is as hard
to approximate as the notoriously-hard densest $k$-subgraph problem. 

While the $k$-forest problem is hard to approximate in the worst-case, we show
that  with the use of resource augmentation, we can efficiently approximate it
up to a constant factor.

First, we restate the problem in terms of the number of demands that are {\em
not} connected.  In particular, the objective of the $k$-forest problem can be
viewed as to remove at most $m-k$ demands and find a minimum-cost subgraph that
connects the remaining demands. We use this perspective of the problem to
explain the performance of our algorithm (in terms of the augmentation) in a
more intuitive way.

Specifically, we present a polynomial-time algorithm for the $k$-forest problem
that, for every $\epsilon>0$, removes at most $m-k$ demands and has cost no
more than $O(1/\epsilon^{2})$ times the cost of an optimal algorithm that
removes at most $(1-\epsilon)(m-k)$ demands.  
\end{abstract}

%% file: intro.tex
\section{Introduction}\label{sec:intro}
\vspace{-2pt}
In the worst-case paradigm, algorithms for NP-hard problems are typically
characterized by their \emph{approximation ratio}, defined as the ratio between
the worst-case cost of the algorithm and the cost of an all-powerful optimal
algorithm.  Many computationally-hard problems admit efficient
worst-case approximations \cite{johnson1974approximation,
klein2010approximation, %coffman2013bin, hochbaum1996approximation,
williamson2011design, vazirani2013approximation}.  However, there are several 
fundamental problems, such as $k$-densest
subgraph~\cite{FeigePeleg01:The-dense-k-subgraph,
BhaskaraCharikar10:Detecting-high}, %FeigeSeltser97:densest-k-subgraph}, 
set cover~\cite{lund1994hardness, feige1996threshold}, graph
coloring~\cite{blum1994new, wigderson1983improving, blum1997o}, 
%lund1994hardness, feige1998zero} 
etc., for which no algorithm with a \emph{reasonable}
approximation guarantee is known.  

Many problems that are hard in the worst-case paradigm admit simple and fast
heuristics in practice.  Illustrative examples include clustering problems (e.g. $k$-median, $k$-means and correlation clustering) and SAT
problems---simple algorithms and solvers for these NP-hard problems routinely
find meaningful clusters~\cite{DanielyLinial12:Clustering-is-difficult} and
satisfiable solutions~\cite{OhrimenkoStuckey09:Propagation-via-lazy} on
practical instances respectively.  A major direction in algorithmic research is
to explain the gap between the observed practical performance and the provable
worst-case guarantee of these algorithms.  Previous work has looked at various
approaches to analyze algorithms that rules out pathological
worst-cases~\cite{koutsoupias2000beyond,
borodin1995competitive,young2000line,emek2009online}.  One such widely-used
approach, especially in the areas of online scheduling and
matching~\cite{chung2008online,PhillipsStein02:Optimal-time-critical,KalyanasundaramPruhs00:Speed-is-as-powerful,kalyanasundaram2000online},
%anthony2014online},
is the model of {\em resource augmentation}.

In the resource-augmentation model, an algorithm is given some additional power
and its performance is compared against that of an optimal algorithm without
the additional power. Resource augmentation has been studied in various guises
such as speed augmentation and machine augmentation (see~\secref{related}
for details). 
Recently, \citet{LucarelliThang16:Online-Non-preemptive} unified
the different notions of resource augmentation under a \emph{generalized
resource-augmentation model} that is based on LP duality.  Roughly speaking, in
the generalized resource-augmentation model, the performance of an algorithm is
measured by the ratio between its worst-case objective value over the set of
feasible solutions $\mathcal{P}$ and the optimal value which is constrained
over a set $\mathcal{Q}$ that is a \emph{strict} subset of $\mathcal{P}$.  In
other words, in the unified model, the algorithm is allowed to be optimized
over relaxed constraints while the adversary (optimum) has tighter constraints.

Duality-based techniques have proved to be powerful tools in the area of online
scheduling with resource augmentation.  Starting with the seminal work
of~\cite{AnandGarg12:Resource-augmentation}, many competitive algorithms have
been designed for online scheduling
problems~\cite{%AnandGarg12:Resource-augmentation,
DevanurHuang14:Primal-Dual,
%GuptaKrishnaswamy12:Online-Primal-Dual,
ImKulkarni14:Competitive-Algorithms,
ImKulkarni14:SELFISHMIGRATE:-A-Scalable,
ImKulkarni15:Competitive-Flow,
Thang13:Lagrangian-Duality,
angelopoulos2015primal,
LucarelliThang16:Online-Non-preemptive}.
%The main objective of this line of work has been to design efficient \emph{scalable} algorithms, that is, 
%algorithms that violate some constraint by only a ``small''($1+\epsilon$) factor
%and achieve a constant approximation ratio (depending on $\epsilon$), for any given $\epsilon > 0$.  
%
Interestingly, the principle ideas behind the duality-based approach in the
resource-augmentation setting are general and can be applied to other
(non-scheduling, offline) optimization problems as well. 

In this paper, we
initiate the use of duality to analyze approximation algorithms with resource
augmentation in the context of general optimization problems. We exemplify this
approach by focusing on a problem that has no reasonable approximation in the
worst-case paradigm---the {\em $k$-forest
problem}~\cite{HajiaghayiJain06:The-prize-collecting-generalized}.  

\pparagraph{The $k$-Forest Problem} In the $k$-forest problem, given an edge-weighted graph $G(V,E)$, a parameter $k$ and
a set of $m$ demand pairs $\subseteq V \times V$, we need to find a
minimum-cost subgraph that connects at least $k$ demand pairs.

The $k$-forest problem is a generalization of the classic $k$-MST (minimum
spanning tree) and the $k$-Steiner tree (with a common source) problems, both of
which admit constant factor approximations. In particular, $k$-MST and
$k$-Steiner tree can be approximated up to a factors of $2$ and $4$
respectively~\cite{chudak2001approximate,garg19963}.  On the other hand, the
$k$-forest problem has resisted similar attempts---the best-known approximation
guarantee is $O(\min\{\sqrt{n},
\sqrt{k}\})$~\cite{GuptaHajiaghayi10:Dial-a-ride-from-k-forest}.

\citet{HajiaghayiJain06:The-prize-collecting-generalized} show that the
$k$-forest problem is roughly as hard as the celebrated {\em densest
$k$-subgraph problem}. Given a graph $G$ and a parameter $k$, the densest
$k$-subgraph problem seeks to find a set of $k$ vertices which induce the
maximum number of edges. 
The densest $k$-subgraph problem has been studied extensively in the
literature~\cite{FeigePeleg01:The-dense-k-subgraph, Khot06:Ruling-out-PTAS,
BhaskaraCharikar10:Detecting-high, AsahiroIwama00:greedy-densest-k,
SrivastavWolf98:Finding-dense-subgraph-SDP,
FeigeLangberg01:Maximization-problem-graph-partitioning,
%FeigeSeltser97:densest-k-subgraph, 
BirnbaumGoldman09:greedy-densest-k-subgraph}
and is regarded to be a hard problem.
\citet{HajiaghayiJain06:The-prize-collecting-generalized} show that if there
is a polynomial time $r$-approximation for the $k$-forest problem then there
exists a polynomial time $2r^{2}$-approximation algorithm for the densest
$k$-subgraph problem. The best known approximation guarantee for the densest
$k$-subgraph problem is $O(n^{1/4+\epsilon})$ %for every $\epsilon > 0$
\cite{BhaskaraCharikar10:Detecting-high}. As pointed out by Hajiaghayi and
Jain~\cite{HajiaghayiJain06:The-prize-collecting-generalized}, an approximation
ratio better than  $O(n^{1/8})$ for the $k$-forest problem (which implies an
approximation ratio better than $O(n^{1/4})$ for the densest $k$-subgraph problem)
would require significantly new insights and techniques.

\subsection{Our Approach and Contributions} 
\vspace{-2pt}
We give the first polynomial-time constant-factor algorithm for the $k$-forest problem in the resource-augmentation model.

Our algorithm is based
on the primal-dual algorithm
by~\citet{HajiaghayiJain06:The-prize-collecting-generalized} for a
closely-related problem, the {\em prize collecting generalized Steiner tree}
(PCGST) problem.  As noted
by~\citet{HajiaghayiJain06:The-prize-collecting-generalized}, the $k$-forest
problem is a Lagrangian relaxation of the PCGST problem.
The authors give a
$3$-approximation algorithm for the PCGST problem. However, their algorithm is
not {\em Lagrangian-multiplier preserving}~\cite{williamson2011design},
which makes it difficult to derive a constant-factor approximation for 
the $k$-forest problem. In this paper, we overcome the challenge posed by the
non-Lagrangian-multiplier-preserving nature of the
primal-dual algorithm by~{\citet{HajiaghayiJain06:The-prize-collecting-generalized}, to obtain
a constant-factor approximation for the $k$-forest problem, by using resource augmentation.

The primal-dual approach is particularly well-suited to analyze algorithms with
resource augmentation. In particular, the
resource augmentation setting can be viewed as a game between an algorithm and
the optimal (or the adversary) where the adversary is subject to tighter
constraints. To apply this notion to the $k$-forest problem, we
need a constraint to play this game between the algorithm and the adversary.  
A natural approach is to choose the number of demands connected as the comparative
constraint. That is, the algorithm chooses to connect at least $k$ ``cheap'' demands
out of the total $m$ demands while the adversary's requirement is higher---to connect 
slightly more than $k$ demands.
An alternate approach is to constrain the number of demands that each algorithm
is allowed to ignore or remove, that is, the algorithm can remove up to $m-k$ ``costly'' demands
while the adversary can remove slightly fewer demands.
Note that with respect to exact and approximate solutions (without any resource augmentation), 
both approaches are equivalent.

We were able to utilize the framework of PCGST~\cite{HajiaghayiJain06:The-prize-collecting-generalized}
and obtain our result by choosing the number of demands {\em that can be removed} as the constraint to be 
augmented.
In particular,
our algorithm for the $k$-forest problem can remove
up to $m-k$ demands whereas the adversary can only remove up to $\lfloor
(1-\epsilon)(m-k) \rfloor$ demands. This tighter cardinality
constraint allows the dual to ``raise'' an additional amount (depending on
$\epsilon$) to ``pay'' for the primal cost.  We exploit this property to bound
the cost of the algorithm's output and that of a dual feasible solution to
derive the approximation ratio. We show the following.

\begin{theorem}\label{thm:k-forest} 
There exists a polynomial-time algorithm for the $k$-forest problem that,
for any $\epsilon > 0$, removes at most
$(m - k)$ connection demands and outputs a subgraph with cost at most
$O(1/\epsilon^{2})$ times the cost of the subgraph output by the optimal algorithm that removes at most
$\lfloor(1-\epsilon)(m - k)\rfloor$ demands.  
\end{theorem}

\pparagraph{Augmentation Parameter: Demands Removed vs. Demands Connected}
While the approach of connecting at least $k$ demands 
is equivalent to rejecting up to $m-k$ demands 
with respect to exact and approximate solutions (without resource augmentation),
there is a notable distinction between them in the presence of augmentation.
In particular, allowing the adversary to remove up to $(1-\epsilon)(m-k)$ demands (compared to $m-k$ demands 
removed by the algorithm), means we require the adversary to connect at least $k + \epsilon (m-k)$ demands (compared to the $k$ demands connected by 
the algorithm).

In this paper, we provide augmentation in terms of $m-k$, the number of demands that can
be removed, because it leads to a more intuitive
understanding of our algorithm's performance. 
In particular, our algorithm is {\em scalable} in terms of the parameter $m-k$,
that is, it is a constant-factor approximation (depending on $\epsilon$) with a
factor $(1+\epsilon)$ augmentation. 
On the other hand, in terms of the parameter $k$, our algorithm is a constant-factor approximation (depending on $\epsilon$)
with a factor $\left(1+\frac{m-k}{k} \cdot \epsilon\right)$ augmentation, which is arguably not as insightful. 
We leave the question
of obtaining a constant-factor approximation with a better augmentation in
terms of $k$ as an interesting open problem.

\subsection{Additional Related Work}\label{sec:related} 

\pparagraph{$k$-Forest and Variants} The $k$-forest problem generalizes both
$k$-MST and $k$-Steiner tree.  \citet{chudak2001approximate} discuss the
$2$-approximation for $k$-MST~\cite{garg19963} and give a $4$-approximation for
$k$-Steiner tree. %Note that . 
  
\citet{SegevSegev10:Approximate-k-Steiner} gave a $O(\min\{n^{2/3},
\sqrt{m}\}\log n)$-approximation for the $k$-forest problem, which was improved
by~\citet{GuptaHajiaghayi10:Dial-a-ride-from-k-forest} to a $O(\min{\sqrt{n},
\sqrt{k}})$-approximation. \citet{GuptaHajiaghayi10:Dial-a-ride-from-k-forest}
also reduce a well-studied vehicle-routing problem
in operations research, the {\em Dial-a-Ride}
problem~\cite{CharikarRaghavachari98, %KrumkeRambau00,
HaimovichRinnooy85,FredericksonHecht76}
to the $k$-forest problem.  In particular, they show that an
$\alpha$-approximation for $k$-forest implies an $O(\alpha
\log^2n)$-approximation algorithm for the Dial-a-Ride problem.

\pparagraph{Resource Augmentation and Duality}
\citet{KalyanasundaramPruhs00:Speed-is-as-powerful} initiated the study of
resource augmentation with the notion of \emph{speed augmentation}, where an
online scheduling algorithm is compared against an adversary with slower
processing speed.  \citet{PhillipsStein02:Optimal-time-critical} proposed the
\emph{machine augmentation} model in which the algorithm has more machines than
the adversary.  \citet{ChoudhuryDas15:Rejecting-jobs} introduced the
\emph{rejection model} where an online scheduling algorithm is allowed to
discard a small fraction of jobs.  Many natural scheduling algorithms can be
analyzed using these models and these analyses have provided theoretical
evidence behind the practical performance of several scheduling heuristics.
Recently, \citet{LucarelliThang16:Online-Non-preemptive} unified the different notions under
a \emph{generalized resource-augmentation model} using LP duality. 
To the best of our knowledge, such duality-based techniques have not been used in the context of 
approximation algorithms with resource augmentation.

%% file: k-forest.tex
\section{Primal-Dual Algorithm for $k$-Forest} \label{k-forest}

In this section, we present an efficient primal-dual algorithm for the $k$-forest problem in the resource-augmentation model.

In the $k$-forest problem, given an undirected graph $G(V,E)$ with  a
nonnegative cost $c_{e}$ on each edge $e \in E$, a parameter $k$, and $m$
connection demands $\mathcal{J} = \{(s_1, t_1), (s_2, t_2), \ldots,
(s_{m},t_{m})\} \subseteq V \times V$, the objective is to 
construct a minimum-cost subgraph of $G$ which connects at least $k$ demands.
To overcome the non-Lagrangian-multiplier-preserving barrier~\cite{HajiaghayiJain06:The-prize-collecting-generalized} and
to take advantage of resource augmentation, we restate the problem
as follows---given an undirected graph $G(V,E)$ with  a
nonnegative cost $c_{e}$ on each edge $e \in E$, a parameter $k$, and $m$
connection demands $\mathcal{J} = \{(s_1, t_1), (s_2, t_2),\ldots, 
(s_{m},t_{m})\} \subseteq V \times V$, the objective is remove up to $(m-k)$ demands 
and construct a minimum-cost subgraph of $G$ that connects 
the remaining demands.

We use the algorithm by \citet{HajiaghayiJain06:The-prize-collecting-generalized}
for the prize-collecting generalized Steiner tree (PCGST) problem 
and refer to it by the shorthand \HJ. In the prize-collecting generalized Steiner tree (PCGST) 
problem, given an undirected graph $G(V,E)$,  with  a nonnegative cost $c_{e}$ on each edge $e \in E$, 
$m$ connection demands $\mathcal{J} = \{(s_1, t_1), (s_2, t_2), \ldots,
(s_{m},t_{m})\}$ and a nonnegative penalty cost $\pi_i$ for every demand $i \in \mathcal{J}$,
the goal is minimize the cost of buying a set of edges and paying a penalty for the demands
that are not connected by the chosen edges. Without loss of generality, we can 
assume that $\mathcal{J} \subset V \times V$, as the penalty for demands that need
not be connected can be set to zero.

Next, we restate the LP for the PCGST problem in terms of the $k$-forest problem and 
reproduce the relevant lemmas~\cite{HajiaghayiJain06:The-prize-collecting-generalized}. 

\subsection{Hajiaghayi and Jain's LP for $k$-forest}

Fix a constant $0 < \epsilon < 1$. Set $\tepsilon = \epsilon/2$ and set $r = (1-\tepsilon)(m-k)$.
Let $x_e$ be a variable such that $x_e = 1$ if edge $e \in E$ is included in the subgraph solution. Similarly, let $z_{i}$
be a variable such that $z_{i} = 1$ if $s_{i}, t_{i}$ are not connected in the subgraph solution. 
We restate the integer program for the PCGST problem~\cite{HajiaghayiJain06:The-prize-collecting-generalized} 
in terms of the $k$-forest problem in the resource augmentation model as $(\mathcal{P}_{\tepsilon})$. 
\begin{align}
& & \min  \sum_{e \in E} & c_e x_e  	& & (\mathcal{P}_{\tepsilon}) \notag \\
(y_S) & &	\qquad  \sum_{e \in \delta(S)} x_{e} + z_{i} &\ge 1	& &  \forall i, \forall S \subset V: S \odot i \notag\\
(\lambda) & &	 \qquad \sum_{i,j \in V} z_{i}& \leq (1-\tepsilon) r & &  \notag\\
& & x_e, z_{i} &\in \{0,1\} & & \forall e \in E, \forall i \notag 
\end{align}
For a set $S \subset V$, the notation $S \odot i$ stands for $|\{s_{i},t_{i}\} \cap S| =1$. 
For a given non-empty set $S \subset V$,  $\delta(S)$ denotes the set of edges defined by the {\em cut} $S$, that is,
$\delta (S)$ is the set of all edges with exactly one endpoint in $S$. Thus, the first constraint says that for every cut $S \odot i$,
there is at least one edge $e \in \delta(S)$ such that either edge $e$ is included in the solution or demand $i$
is removed.  The second constraint says that the total number of demands removed is no more than 
$(1-\tepsilon)r$. Note that the optimal value of $(\mathcal{P}_{\tepsilon})$ is 
a lower bound on the optimal solution that removes at most $(1-\epsilon)(m-k)$ demands. 
This is because we have slightly relaxed the upper bound of the number of demands removed to be 
$(1-\tepsilon)r = (1-\tepsilon)^2r \geq (1-\epsilon)(m - k)$. 

The dual $(\mathcal{D}_{\tepsilon})$ of the relaxation of $(\mathcal{P}_{\tepsilon})$ follows.
\begin{align*}
\max  \sum_{S \subset V, S \odot i} y_{i,S} & -  (1-\tepsilon)r\lambda 	& & (\mathcal{D}_{\tepsilon}) \notag \\
\sum_{S: e \in \delta(S), S \odot i} y_{i,S}  &\le  c_e	& &  \forall e \in E \notag\\
\sum_{S: S \odot i} y_{i,S} & \le \lambda & &  \forall i \notag\\
y_{i,S} &\ge 0 & & \forall S \subset V: S \odot i \notag 
\end{align*}
Hajiaghayi and Jain~\cite{HajiaghayiJain06:The-prize-collecting-generalized} formulate a new dual $(\mathcal{D}_{\tepsilon}^{ \HJ})$ 
equivalent to  $(\mathcal{D}_{\tepsilon})$ based
on Farkas lemma. This new dual resolves the challenges posed by raising different dual variables associated with the same set of vertices
of the graph in $(\mathcal{D}_{\tepsilon})$. 
We refer the readers to the original paper~\cite{HajiaghayiJain06:The-prize-collecting-generalized}
for a detailed discussion on the transformation and proofs.

Note that $\mathcal{S}$ is a \emph{family} of subsets of $V$ if $\mathcal{S} = \{S_1, S_2, \ldots, S_\ell\}$ 
where $S_{j} \subset V$ for $1 \leq j \leq \ell$. For a family $\mathcal{S}$, if there exists
$S \in \mathcal{S}$ such that $S \odot i$, we denote it by $\mathcal S \odot i$. The new dual $(\mathcal{D}_{\tepsilon}^{\HJ})$
is stated below.
\begin{align*}
\max  \sum_{S \subset V} y_S & - (1-\tepsilon) r\lambda  	& & (\mathcal{D}_{\tepsilon}^{ \HJ}) \notag \\
\sum_{S: e \in \delta(S)} y_{S}  &\le  c_e	& &  \forall e \in E \notag\\
\sum_{S \in \mathcal{S}} y_{S} & \le \sum_{i, \mathcal{S} \odot i}\lambda & &  \forall \mbox{ family $\mathcal{S}$} \notag\\
y_{S} &\ge 0 & & \forall S \subset V \notag 
\end{align*}
We use the \HJ~algorithm (along with 
the construction of dual variables) for the PCGST problem. We set the penalty of every request 
to a fixed constant $\lambda$. We reproduce the relevant lemmas in terms of $k$-forest. See~\cite{HajiaghayiJain06:The-prize-collecting-generalized} for proofs.

For $S \subset V$, let $y_{S}(\lambda)$'s be the dual variables constructed in \HJ~algorithm with penalty cost $\lambda$. 
Let $y(\lambda)$ be the vector consisting of all $y_{S}(\lambda)$'s.

\begin{lemma}[\cite{HajiaghayiJain06:The-prize-collecting-generalized}]\label{lem:penalty}
Let $r(\lambda)$ be the number of demands removed with the penalty cost $\lambda$ by the 
\HJ~algorithm. Then, $r(\lambda) \cdot \lambda \leq \sum_{S} y_{S}(\lambda)$.
\end{lemma} 

\begin{lemma}[\cite{HajiaghayiJain06:The-prize-collecting-generalized}]\label{lem:primal-bound}
Let $F$ be the set of edges in the subgraph solution output by the \HJ~algorithm. 
Then $\sum_{e \in F} c_{e} \leq 2 \sum_{S} y_{S}(\lambda)$.
\end{lemma}

\subsection{Algorithm for $k$-Forest}\label{sec:binary}
Let $\HJ(\lambda)$ denote a call to the primal-dual algorithm of Hajiaghayi and Jain~\cite{HajiaghayiJain06:The-prize-collecting-generalized} 
for the PCGST problem with a penalty cost $\lambda$ for every request. 
For a given value $\lambda$, let $r(\lambda)$ be the number
of demands removed by the algorithm $\HJ(\lambda)$.
Similar to the classic $k$-median algorithm~\cite{jain2001approximation}, 
we do a binary search on the value of $\lambda$, and call the \HJ~as a subroutine each time.  
We describe our algorithm for $k$-forest next and refer to it as algorithm $\mathcal{A}$.

\begin{enumerate}
\item Let $c_{\min} = \min\{c_{e}: e \in E\}$. Initially set $\lambda^{1} \gets 0$ and $\lambda^{2} \gets {\sum_{e \in E} c_{e}}$. 
\item While $(\lambda^{2} - \lambda^{1}) > c_{\min}/m^{2}$, do the following:
\begin{enumerate}%[noitemsep,nolistsep,leftmargin=*]
	\item Set $\lambda =  (\lambda^{1} + \lambda^{2})/2$. 
	\item Call $\HJ(\lambda)$ and get $r(\lambda)$ (the number of demands removed).
	\begin{enumerate}
		\item\label{step1} If $r(\lambda) =r $, then output the solution given by $\HJ(\lambda)$.
		\item Otherwise, if $r(\lambda) < (1-\epsilon/2)r$ then update $\lambda^{2} \gets \lambda$;
		\item Otherwise, if $r(\lambda) > r$ then update $\lambda^{1} \gets \lambda$.
	\end{enumerate}
\end{enumerate}
\item\label{step3} 
Let $\alpha_1$ and $\alpha_2$ be such that $\alpha_1 r_1 + \alpha_2 r_2 = r$, $\alpha_1 + \alpha_2 =1$ and $\alpha_1, \alpha_2 \ge 0$.  Specifically, 
\begin{equation}\label{eq:alpha-forest}
\alpha_1 = \frac{r  - r_2}{r_1 - r_2} \mbox{\hspace{5pt} and \hspace{5pt}} \alpha_2
= \frac{r_1 - r_{0}}{r_1 - r_2}
\end{equation}
If $\alpha_{2} \geq \tepsilon$, then return the solution $\HJ(\lambda^{2})$. Else,
return the solution $\HJ(\lambda^{1})$.
\end{enumerate}

Observe that the algorithm $\mathcal{A}$ always terminates: either it encounters a value of $\lambda$ such that 
$r(\lambda) = r$ in Step \ref{step1} or returns a solution depending on the final values of $\lambda^{1}$ and $\lambda^{2}$ in Step~\ref{step3}.

\subsection{Analysis}\label{sec:k-forest-analysis}
Let $\OPT_{u}$ be the cost of an optimal solution that removes at most $u$ demands.
Assume that $c_{\min} \le \OPT_{(1-\tepsilon)r}$, 
because otherwise the optimal solution is to not select any edge $e \in E$.
The algorithm outputs the solution either in Step~\ref{step1} or in Step~\ref{step3}. 
First, consider the case that the solution is output in Step~\ref{step1}.

\begin{lemma}\label{lem:step2-analysis}
Suppose that $\mathcal{A}$ outputs the solution given by $\HJ(\lambda)$ in Step~\ref{step1} for some $\lambda$.
Let $F$ be the set of edges returned by $\HJ(\lambda)$. 
Then, 
$$
\sum_{e \in F} c_{e} \leq  \frac{2}{\tepsilon} \cdot \OPT_{(1-\tepsilon)r}.
$$ 
\end{lemma}
\begin{proof}
Since the solution is output in Step~\ref{step1}, the number of demands removed is $r(\lambda) = r$. 
By weak duality, the value of $\OPT_{(1-\tepsilon)r}$ is lower bounded by the 
objective cost of $(\mathcal{D}_{\tepsilon}^{ \HJ})$ with dual variables $y(\lambda)$. 
That is, 
\begin{align*}
\OPT_{(1-\tepsilon)r} &\geq \sum_{S \subset V} y_S - (1-\tepsilon) r\lambda\\ 
&\geq  \tepsilon \cdot \sum_{S \subset V} y_S \geq  \frac{\tepsilon}{2} \cdot \sum_{e \in F} c_{e}
\end{align*}
where the last two inequalities follow from~\lemref{penalty} and~\ref{lem:primal-bound} respectively.\qed
\end{proof}

Next, consider the case that the solution is output in Step~\ref{step3}. 
Let $F_1$ and $F_2$ be the sets of edges returned by  $\HJ(\lambda^{1})$ and $\HJ(\lambda^{2})$, respectively. 
Let $r_1$ and $r_{2}$ denote the number of demands removed by 
$\HJ(\lambda^{1})$ and $\HJ(\lambda^{2})$ respectively. Then, we have $\lambda^{2} - \lambda^{1} \leq c_{\min}/m^{2}$.
As $c_{\min} \le \OPT_{(1-\tepsilon)r}$, at the end of the while loop we have 
$\lambda^{2} - \lambda^{1} \leq c_{\min}/m^{2} \le \OPT_{(1-\tepsilon)r}/m^{2}$.
Furthermore, $ r_2 < r <r_1$.

Consider the dual vector $(y^*, \lambda^*)$ defined as 
$$(y^*, \lambda^*) = \alpha_1  (y(\lambda_1), \lambda_1 ) 
+ \alpha_2 (y(\lambda_2), \lambda_2 )$$
where the coefficients $\alpha_{1}$ and $\alpha_{2}$ are defined in Step~\ref{step3} of algorithm $\mathcal{A}$.
Then, $(y^*, \lambda^*)$ forms a feasible solution to the dual $(\mathcal{D}_{\tepsilon}^{\HJ})$
as it is a convex combination of two dual feasible solutions. 

We bound the cost of algorithm~$\mathcal{A}$ by bounding the cost of the
dual $(\mathcal{D}_{\tepsilon}^{\HJ})$. 
\begin{lemma}\label{lem:main-k} 
$\alpha_{1} \sum_{e \in F_{1}} c_{e} +  \alpha_{2} \sum_{e \in F_{2}} c_{e} \le \frac{4}{\tepsilon} \cdot \OPT_{(1-\tepsilon)r}.$
\end{lemma}
\begin{proof}
The cost of the dual $(\mathcal{D}_{\tepsilon}^{\HJ})$ lower bounds the cost of an optimal algorithm that removes at most $(1-\tepsilon)r$ demands.  
That is,
\begin{align}
&  \OPT_{(1-\tepsilon)r} \ge   \biggl( \sum_{S} y_{S}(\lambda^*) - (1-\tepsilon) r \lambda^* \biggr) \notag \\
&= \alpha_{1}  \biggl( \sum_{S} y_{S}(\lambda_{1}) - (1-\tepsilon) r_{1} \lambda^* \biggr) +
\alpha_{2} \biggl( \sum_{S} y_{S}(\lambda_{2}) - (1-\tepsilon) r_{2} \lambda^* \biggr) \notag\\
&= \alpha_{1}  \biggl( \sum_{S} y_{S}(\lambda_{1}) - (1-\tepsilon) r_{1} \lambda_1 \biggr) -\alpha_1 (1-\tepsilon) r_1 (\lambda^*-\lambda_1) \notag\\
& \qquad 
 +\alpha_{2} \biggl( \sum_{S} y_{S}(\lambda_{2}) - (1-\tepsilon) r_{2} \lambda_2 \biggr) + \alpha_2 (1-\tepsilon) r_2 (\lambda_2-\lambda^*) \notag\\
&\ge \alpha_{1}  \biggl( \sum_{S} y_{S}(\lambda_{1}) - (1-\tepsilon) r_{1} \lambda_1 \biggr)
 +\alpha_{2} \biggl( \sum_{S} y_{S}(\lambda_{2}) - (1-\tepsilon) r_{2} \lambda_2 \biggr) - m (\lambda^*-\lambda_1)\label{ineq: lambda}\\
%&\ge \alpha_{1}  \biggl( \sum_{S} y_{S}(\lambda_{1}) - (1-\tepsilon) r_{1} \lambda_1 \biggr)
% +\alpha_{2} \biggl( \sum_{S} y_{S}(\lambda_{2}) - (1-\tepsilon) r_{2} \lambda_2 \biggr) - \frac{\OPT_{(1-\tepsilon)r}}{m} \label{ineq: opt}\\
&\ge \tepsilon \biggl[ \alpha_{1} \frac 1\tepsilon \biggl( \sum_{S} y_{S}(\lambda_{1}) - (1-\tepsilon) r_{1} \lambda_1 \biggr) \notag\\
& \qquad \qquad \qquad +\alpha_{2} \frac 1\tepsilon \biggl( \sum_{S} y_{S}(\lambda_{2}) - (1-\tepsilon) r_{2} \lambda_2 \biggr) \biggr] 
- \frac{\OPT_{(1-\tepsilon)r}}{m} \label{ineq: opt}\\
&= \tepsilon \alpha_{1} \biggl[  \biggl(\frac 1\tepsilon -1\biggr) \biggl( \sum_{S} y_{S}(\lambda_{1}) - r_{1} \lambda_1 \biggr) + 
 \sum_S y_S (\lambda_1)\biggr] \notag\\
& \qquad \qquad \qquad 
   + \tepsilon \alpha_{2} \biggl[ \biggl(\frac 1\tepsilon -1\biggr) \biggl( \sum_{S} y_{S}(\lambda_{2}) - r_{2} \lambda_2 \biggr) +
\sum_S y_S (\lambda_2) \biggr] - \frac{\OPT_{(1-\tepsilon)r}}{m} \notag\\
&\ge \tepsilon \biggl( \alpha_{1} \sum_S y_S (\lambda_1) + \alpha_2 \sum_S y_S (\lambda_2) \biggr) - \frac{\OPT_{(1-\tepsilon)r}}{m} \label{ineq: pos}\\
&\ge \frac \tepsilon 2 \biggl( \alpha_{1} \sum_{e \in F_1} c_e + \alpha_2 \sum_{e \in F_2} c_e  \biggr) - \frac{\OPT_{(1-\tepsilon)r}}{m} \label{ineq:uselemma}
\end{align}
Inequality~(\ref{ineq: lambda}) holds because $\lambda_1 \leq \lambda^* \leq \lambda_2$,  $r_{1} < m$, $0 \leq \alpha_{1}, \alpha_{2} \leq 1$
and $0 < \tepsilon < 1$.
Inequality~(\ref{ineq: opt}) follows from the definition of the penalty costs, that is, $\lambda^{*} - \lambda_{1} \leq \lambda_{2} - \lambda_{1} \leq \OPT_{(1-\tepsilon)r}/m^{2}$.
Inequality~(\ref{ineq: pos}) follows from \lemref{penalty} and the fact that $1/\tepsilon -1 >0$. Finally, Inequality~(\ref{ineq:uselemma}) uses \lemref{primal-bound}.

Rearranging the terms of Inequality~(\ref{ineq:uselemma}) proves~\lemref{main-k}, that is, 
\begin{align*}
\alpha_{1} \sum_{e \in F_{1}} c_{e} +  \alpha_{2} \sum_{e \in F_{2}} c_{e} \le   \frac 2 \tepsilon \cdot \frac{m+1}{m} \cdot 
 \OPT_{(1-\tepsilon)r} \le  \frac 4 \tepsilon  \cdot \OPT_{(1-\tepsilon)r}.
\end{align*}
\qed
\end{proof}

We are now ready to prove the main theorem. %, restated below.

\begin{proofof}{\thmref{k-forest}}
We analyze algorithm $\mathcal{A}$. \lemref{step2-analysis} is sufficient for the case that $\mathcal{A}$
outputs the solution in Step~\ref{step1}. Now suppose that $\mathcal{A}$ outputs
the solution in Step~\ref{step3}.

Note that $(1-\tepsilon)r \geq (1-\epsilon)(m-k) \geq \lfloor (1-\epsilon)(m - k) \rfloor$,
therefore, we have,
\[\OPT_{(1-\tepsilon)r} \leq \OPT_{\lfloor (1-\epsilon)(m - k) \rfloor}.\]
We consider two cases based on the value of $\alpha_{2}$.\\
\pparagraph{Case 1:  $\alpha_2 \ge \tepsilon$} $\mathcal{A}$ returns $F_2$ which is a feasible solution since 
the number of demands removed is $r_2 \le r$. We bound the cost of solution $F_2$ using~\lemref{main-k}:
\begin{align*}	%\label{eq:case1}
 \sum_{e \in F_2} c_e &\le \frac 1\tepsilon \alpha_2   \sum_{e \in F_2} c_e \le \frac 1\tepsilon \biggl( \alpha_1  \sum_{e \in F_1} c_e + \alpha_2  \sum_{e \in F_2} c_e \biggr) \\
&\le  \frac{4}{\tepsilon^2}\cdot \OPT_{(1-\tepsilon)r}
\le \frac{4}{\tepsilon^2}\cdot \OPT_{\lfloor (1-\epsilon)(m - k) \rfloor}.
\end{align*}
\pparagraph{Case 2: $\alpha_2 < \tepsilon$} $\mathcal{A}$ outputs $F_{1}$ as the solution. Since
$\alpha_1 + \alpha_2 =1$ by definition, we have $\alpha_1 > 1-\tepsilon$. 
Using equation~(\ref{eq:alpha-forest}), we have:
\begin{align*}
r - r_2 \ge (1-\tepsilon) (r_1 - r_2) 
\Rightarrow  r - \tepsilon r_2 \ge (1-\tepsilon) r_1 
\Rightarrow r_1 \le \frac{1}{(1-\tepsilon)} \cdot r = (m - k)
\end{align*}
where the last equality uses $r = (1-\tepsilon)(m-k)$. Thus, 
$F_{1}$ is a feasible solution.

We bound the cost of solution $F_1$, applying~\lemref{main-k} again:
\begin{align*}	%\label{eq:case2}
 \sum_{e \in F_1} c_e &\le \frac{1}{1-\tepsilon} \alpha_1 \sum_{e \in F_1} c_e 
 	\le  \frac{1}{1-\tepsilon} \biggl( \alpha_1  \sum_{e \in F_1} c_e + \alpha_2  \sum_{e \in F_2} c_e \biggr) \\
	&\le  \frac{4}{\tepsilon^{2}}\cdot \OPT_{(1-\tepsilon)r}
	\le \frac{4}{\tepsilon^{2}}\cdot \OPT_{\lfloor (1-\epsilon)(m - k) \rfloor}
\end{align*}
where the third inequality holds since $(1 - \tepsilon) \geq 1/2 \geq \tepsilon$. 

The two cases together prove the approximation and augmentation factors of~$\mathcal{A}$~in~\thmref{k-forest} (recall that $\tepsilon = \epsilon/2$).

$\mathcal{A}$ makes $O\bigl( \log( \frac{1}{\epsilon} m^{2} \frac{\sum_{e} c_{e}}{c_{\min}}) \bigr)$ calls to the polynomial-time \HJ~algorithm.
Thus, $\mathcal{A}$'s running time is polynomial in the size of the input and $\log 1/\epsilon$.  
\qed
\end{proofof}

%% file: conclusion.tex
\section{Conclusion}\label{sec:conclusion}
\vspace*{-2pt}
The model of resource augmentation has been widely-used and has successfully 
provided theoretical evidence for several heuristics, especially in the case of online scheduling problems. 
Surprisingly, for offline algorithms, not many scalable approximation algorithms have been designed,
despite the need of effective algorithms for hard problems.

In this paper, we initiate the study of hard (to approximate) problems in the resource-augmentation model. 
We show that the $k$-forest problem can be approximated up to a constant factor using augmentation.
It is an interesting direction to design algorithms in the resource augmentation model 
for other hard problems which currently admit no meaningful approximation guarantees. 
  

%% file: ack.tex
\subsection*{Acknowledgments}
We thank Samuel McCauley for giving us his valuable feedback.

%% file: main.bbl
\begin{thebibliography}{44}
\providecommand{\natexlab}[1]{#1}
\providecommand{\url}[1]{\texttt{#1}}
\providecommand{\urlprefix}{}

\bibitem[{Anand et~al.(2012)Anand, Garg, and
  Kumar}]{AnandGarg12:Resource-augmentation}
Anand, S., Garg, N., Kumar, A.: Resource augmentation for weighted flow-time
  explained by dual fitting.
\newblock In: Proc. 23rd Symposium on Discrete Algorithms. pp. 1228--1241
  (2012)

\bibitem[{Angelopoulos et~al.(2015)Angelopoulos, Lucarelli, and
  Thang}]{angelopoulos2015primal}
Angelopoulos, S., Lucarelli, G., Thang, N.K.: Primal-dual and dual-fitting
  analysis of online scheduling algorithms for generalized flow time problems.
\newblock In: Proc. 23rd European Symposium on Algorithms. pp. 35--46 (2015)

\bibitem[{Asahiro et~al.(2000)Asahiro, Iwama, Tamaki, and
  Tokuyama}]{AsahiroIwama00:greedy-densest-k}
Asahiro, Y., Iwama, K., Tamaki, H., Tokuyama, T.: Greedily finding a dense
  subgraph.
\newblock Journal of Algorithms 34(2), 203--221 (2000)

\bibitem[{Bhaskara et~al.(2010)Bhaskara, Charikar, Chlamtac, Feige, and
  Vijayaraghavan}]{BhaskaraCharikar10:Detecting-high}
Bhaskara, A., Charikar, M., Chlamtac, E., Feige, U., Vijayaraghavan, A.:
  Detecting high log-densities: an ${O}(n^{1/4})$ approximation for densest
  k-subgraph.
\newblock In: Proc. 42nd Symposium on Theory of Computing. pp. 201--210 (2010)

\bibitem[{Birnbaum and
  Goldman(2009)}]{BirnbaumGoldman09:greedy-densest-k-subgraph}
Birnbaum, B., Goldman, K.J.: An improved analysis for a greedy remote-clique
  algorithm using factor-revealing lps.
\newblock Algorithmica 55(1), 42--59 (2009)

\bibitem[{Blum(1994)}]{blum1994new}
Blum, A.: New approximation algorithms for graph coloring.
\newblock Journal of the ACM 41(3), 470--516 (1994)

\bibitem[{Blum and Karger(1997)}]{blum1997o}
Blum, A., Karger, D.: An {\~o} (n314)-coloring algorithm for 3-colorable
  graphs.
\newblock Information processing letters 61(1), 49--53 (1997)

\bibitem[{Borodin et~al.(1995)Borodin, Irani, Raghavan, and
  Schieber}]{borodin1995competitive}
Borodin, A., Irani, S., Raghavan, P., Schieber, B.: Competitive paging with
  locality of reference.
\newblock Journal of Computer and System Sciences 50(2), 244--258 (1995)

\bibitem[{Charikar and Raghavachari(1998)}]{CharikarRaghavachari98}
Charikar, M., Raghavachari, B.: The finite capacity dial-a-ride problem.
\newblock In: Proc. 39th Symposium on Foundations of Computer Science. pp.
  458--467 (1998)

\bibitem[{Choudhury et~al.(2015)Choudhury, Das, Garg, and
  Kumar}]{ChoudhuryDas15:Rejecting-jobs}
Choudhury, A.R., Das, S., Garg, N., Kumar, A.: Rejecting jobs to minimize load
  and maximum flow-time.
\newblock In: Proc. 26th Symposium on Discrete Algorithms. pp. 1114--1133
  (2015)

\bibitem[{Chudak et~al.(2001)Chudak, Roughgarden, and
  Williamson}]{chudak2001approximate}
Chudak, F.A., Roughgarden, T., Williamson, D.P.: Approximate k-msts and
  k-steiner trees via the primal-dual method and lagrangean relaxation.
\newblock In: Proc. 8th Conference on Integer Programming and Combinatorial
  Optimization. pp. 60--70 (2001)

\bibitem[{Chung et~al.(2008)Chung, Pruhs, and Uthaisombut}]{chung2008online}
Chung, C., Pruhs, K., Uthaisombut, P.: The online transportation problem: On
  the exponential boost of one extra server.
\newblock In: Proc. 8th Latin American Symposium on Theoretical Informatics.
  pp. 228--239 (2008)

\bibitem[{Daniely et~al.(2012)Daniely, Linial, and
  Saks}]{DanielyLinial12:Clustering-is-difficult}
Daniely, A., Linial, N., Saks, M.: Clustering is difficult only when it does
  not matter.
\newblock arXiv preprint arXiv:1205.4891  (2012)

\bibitem[{Devanur and Huang(2014)}]{DevanurHuang14:Primal-Dual}
Devanur, N.R., Huang, Z.: Primal dual gives almost optimal energy efficient
  online algorithms.
\newblock In: Proc. 25th Symposium on Discrete Algorithms (2014)

\bibitem[{Emek et~al.(2009)Emek, Fraigniaud, Korman, and
  Ros{\'e}n}]{emek2009online}
Emek, Y., Fraigniaud, P., Korman, A., Ros{\'e}n, A.: Online computation with
  advice.
\newblock In: Proc. 36th International Colloquium on Automata, Languages, and
  Programming. pp. 427--438 (2009)

\bibitem[{Feige(1996)}]{feige1996threshold}
Feige, U.: A threshold of ln n for approximating set cover (preliminary
  version).
\newblock In: Proc. 28th Symposium on Theory of Computing. pp. 314--318 (1996)

\bibitem[{Feige and
  Langberg(2001)}]{FeigeLangberg01:Maximization-problem-graph-partitioning}
Feige, U., Langberg, M.: Approximation algorithms for maximization problems
  arising in graph partitioning.
\newblock Journal of Algorithms 41(2), 174--211 (2001)

\bibitem[{Feige et~al.(2001)Feige, Peleg, and
  Kortsarz}]{FeigePeleg01:The-dense-k-subgraph}
Feige, U., Peleg, D., Kortsarz, G.: The dense k-subgraph problem.
\newblock Algorithmica 29(3), 410--421 (2001)

\bibitem[{Frederickson et~al.(1976)Frederickson, Hecht, and
  Kim}]{FredericksonHecht76}
Frederickson, G.N., Hecht, M.S., Kim, C.E.: Approximation algorithms for some
  routing problems.
\newblock In: Proc. 17th Symposium on Foundations of Computer Science. pp.
  216--227 (1976)

\bibitem[{Garg(1996)}]{garg19963}
Garg, N.: A 3-approximation for the minimum tree spanning k vertices.
\newblock In: Proc. 37th Symposium on Foundations of Computer Science. pp.
  302--309 (1996)

\bibitem[{Gupta et~al.(2010)Gupta, Hajiaghayi, Nagarajan, and
  Ravi}]{GuptaHajiaghayi10:Dial-a-ride-from-k-forest}
Gupta, A., Hajiaghayi, M., Nagarajan, V., Ravi, R.: Dial a ride from k-forest.
\newblock ACM Transactions on Algorithm 6(2), 41 (2010)

\bibitem[{Haimovich and Rinnooy~Kan(1985)}]{HaimovichRinnooy85}
Haimovich, M., Rinnooy~Kan, A.: Bounds and heuristics for capacitated routing
  problems.
\newblock Mathematics of operations Research 10(4), 527--542 (1985)

\bibitem[{Hajiaghayi and
  Jain(2006)}]{HajiaghayiJain06:The-prize-collecting-generalized}
Hajiaghayi, M.T., Jain, K.: The prize-collecting generalized steiner tree
  problem via a new approach of primal-dual schema.
\newblock In: Proc. 17th Symposium on Discrete Algorithm. pp. 631--640 (2006)

\bibitem[{Im et~al.(2014{\natexlab{a}})Im, Kulkarni, and
  Munagala}]{ImKulkarni14:Competitive-Algorithms}
Im, S., Kulkarni, J., Munagala, K.: Competitive algorithms from competitive
  equilibria: Non-clairvoyant scheduling under polyhedral constraints.
\newblock In: Proc. 46th Symposium on Theory of Computing (2014{\natexlab{a}})

\bibitem[{Im et~al.(2015)Im, Kulkarni, and
  Munagala}]{ImKulkarni15:Competitive-Flow}
Im, S., Kulkarni, J., Munagala, K.: Competitive flow time algorithms for
  polyhedral scheduling.
\newblock In: Proc. 56th Symposium on Foundations of Computer Science. pp.
  506--524 (2015)

\bibitem[{Im et~al.(2014{\natexlab{b}})Im, Kulkarni, Munagala, and
  Pruhs}]{ImKulkarni14:SELFISHMIGRATE:-A-Scalable}
Im, S., Kulkarni, J., Munagala, K., Pruhs, K.: Selfishmigrate: A scalable
  algorithm for non-clairvoyantly scheduling heterogeneous processors.
\newblock In: Proc. 55th Symposium on Foundations of Computer Science
  (2014{\natexlab{b}})

\bibitem[{Jain and Vazirani(2001)}]{jain2001approximation}
Jain, K., Vazirani, V.V.: Approximation algorithms for metric facility location
  and k-median problems using the primal-dual schema and lagrangian relaxation.
\newblock Journal of the ACM 48(2), 274--296 (2001)

\bibitem[{Johnson(1974)}]{johnson1974approximation}
Johnson, D.S.: Approximation algorithms for combinatorial problems.
\newblock Journal of Computer and System Sciences 9(3), 256--278 (1974)

\bibitem[{Kalyanasundaram and
  Pruhs(2000{\natexlab{a}})}]{KalyanasundaramPruhs00:Speed-is-as-powerful}
Kalyanasundaram, B., Pruhs, K.: Speed is as powerful as clairvoyance.
\newblock Journal of the ACM 47(4), 617--643 (2000{\natexlab{a}})

\bibitem[{Kalyanasundaram and
  Pruhs(2000{\natexlab{b}})}]{kalyanasundaram2000online}
Kalyanasundaram, B., Pruhs, K.R.: The online transportation problem.
\newblock SIAM Journal on Discrete Mathematics 13(3), 370--383
  (2000{\natexlab{b}})

\bibitem[{Khot(2006)}]{Khot06:Ruling-out-PTAS}
Khot, S.: Ruling out ptas for graph min-bisection, dense k-subgraph, and
  bipartite clique.
\newblock SIAM Journal on Computing 36(4), 1025--1071 (2006)

\bibitem[{Klein and Young(2010)}]{klein2010approximation}
Klein, P.N., Young, N.E.: Approximation algorithms for {NP}-hard optimization
  problems.
\newblock Chapman \& Hall (2010)

\bibitem[{Koutsoupias and Papadimitriou(2000)}]{koutsoupias2000beyond}
Koutsoupias, E., Papadimitriou, C.H.: Beyond competitive analysis.
\newblock SIAM Journal on Computing 30(1), 300--317 (2000)

\bibitem[{Lucarelli et~al.(2016)Lucarelli, Thang, Srivastav, and
  Trystram}]{LucarelliThang16:Online-Non-preemptive}
Lucarelli, G., Thang, N.K., Srivastav, A., Trystram, D.: Online non-preemptive
  scheduling in a resource augmentation model based on duality.
\newblock In: Proc. 24th European Symposium on Algorithms (2016)

\bibitem[{Lund and Yannakakis(1994)}]{lund1994hardness}
Lund, C., Yannakakis, M.: On the hardness of approximating minimization
  problems.
\newblock Journal of the ACM 41(5), 960--981 (1994)

\bibitem[{Ohrimenko et~al.(2009)Ohrimenko, Stuckey, and
  Codish}]{OhrimenkoStuckey09:Propagation-via-lazy}
Ohrimenko, O., Stuckey, P.J., Codish, M.: Propagation via lazy clause
  generation.
\newblock Constraints 14(3), 357--391 (2009)

\bibitem[{Phillips et~al.(2002)Phillips, Stein, Torng, and
  Wein}]{PhillipsStein02:Optimal-time-critical}
Phillips, C.A., Stein, C., Torng, E., Wein, J.: Optimal time-critical
  scheduling via resource augmentation.
\newblock Algorithmica 32(2), 163--200 (2002)

\bibitem[{Segev and Segev(2010)}]{SegevSegev10:Approximate-k-Steiner}
Segev, D., Segev, G.: Approximate k-steiner forests via the lagrangian
  relaxation technique with internal preprocessing.
\newblock Algorithmica 56(4), 529--549 (2010)

\bibitem[{Srivastav and
  Wolf(1998)}]{SrivastavWolf98:Finding-dense-subgraph-SDP}
Srivastav, A., Wolf, K.: Finding dense subgraphs with semidefinite programming.
\newblock In: Workshop on Approximation Algorithms for Combinatorial
  Optimization. pp. 181--191 (1998)

\bibitem[{Thang(2013)}]{Thang13:Lagrangian-Duality}
Thang, N.K.: Lagrangian duality in online scheduling with resource augmentation
  and speed scaling.
\newblock In: Proc. 21st European Symposium on Algorithms. pp. 755--766 (2013)

\bibitem[{Vazirani(2013)}]{vazirani2013approximation}
Vazirani, V.V.: Approximation Algorithms.
\newblock Springer Science \& Business Media (2013)

\bibitem[{Wigderson(1983)}]{wigderson1983improving}
Wigderson, A.: Improving the performance guarantee for approximate graph
  coloring.
\newblock Journal of the ACM 30(4), 729--735 (1983)

\bibitem[{Williamson and Shmoys(2011)}]{williamson2011design}
Williamson, D.P., Shmoys, D.B.: The design of approximation algorithms.
\newblock Cambridge University Press (2011)

\bibitem[{Young(2000)}]{young2000line}
Young, N.E.: On-line paging against adversarially biased random inputs.
\newblock Journal of Algorithms 37(1), 218--235 (2000)

\end{thebibliography}
